\documentclass[11pt,letterpaper]{article}

\usepackage[utf8]{inputenc}
\usepackage[margin=1.0in]{geometry}
\usepackage{palatino}
\usepackage{macros}
\usepackage{algorithm}
\usepackage{csquotes,fullpage}

\DeclareMathOperator{\EPR}{EPR}

\title{Improved approximation algorithms for the\\ EPR Hamiltonian}
\author{Nathan Ju\thanks{UC Berkeley, \url{nju@berkeley.edu}}\and Ansh Nagda\thanks{UC Berkeley, \url{anshnagda@gmail.com}}}

\begin{document}
\maketitle 
\begin{abstract}
    The EPR Hamiltonian is a family of 2-local quantum Hamiltonians introduced by King \cite{king2023improved}. We introduce a polynomial time $\frac{1+\sqrt{5}}{4}\approx 0.809$-approximation algorithm for the problem of computing the ground energy of the EPR Hamiltonian, improving upon the previous state of the art of $0.72$ \cite{jorquera2024monogamy}. 
    
    As a special case, this also implies a $\frac{1+\sqrt{5}}{4}$-approximation for Quantum Max Cut on bipartite instances, improving upon the approximation ratio of $3/4$ that one can infer in a relatively straightforward manner from the work of Lee and Parekh \cite{eunou}.
\end{abstract}
\thispagestyle{empty}

\newpage

\clearpage
\setcounter{page}{1}

\section{Introduction}
For a quantum Hamiltonian $H$ and a (mixed) state $\rho$, we say that the \emph{energy} of $\rho$ is $\tr(\rho\cdot H)$.
In this work, we refer to a state achieving maximum energy as the \emph{ground state}\footnote{This is a slight deviation from the usual notion of the ground state being the \emph{minimum} energy state.
We remark that both the notions are equivalent up to replacing $H$ by $-H$; we use the maximization notion because it is notationally convenient.}.
Note that the energy of the ground state, or the ground energy, is equal to the maximum eigenvalue of $H$
\[\lambda_{\max}(H) = \max_{\rho\geq 0, \tr(\rho)=1}\tr(\rho\cdot H).\]
While exactly computing the ground energy of local Hamiltonians is computationally infeasible (QMA-hard), \emph{approximating} the ground energy is possibly tractable and is the focus of this work.
Let $\mathcal{F}$ be a family of local Hamiltonians.
For $\alpha\in [0,1]$, an \emph{$\alpha$-approximation algorithm} for $\mathcal{F}$ is one that, upon input $H\in \mathcal{F}$, outputs a description of a (mixed) state $\rho$ achieving energy
\[\tr(\rho\cdot H)\geq \alpha\cdot \lambda_{\max}(H).\]
We are interested in the question ``what is the supremum of all $\alpha\in [0,1]$ such that there is a polynomial time $\alpha$-approximation algorithm\footnote{For the purposes of this paper we only consider classical algorithms, but one can also consider quantum algorithms.} for $\mathcal{F}$?''. We will use $\alpha^*_\mathcal{F}$ to denote the answer to this question.

\paragraph{Quantum Max Cut.}

In this context, the most commonly studied problem is Quantum Max Cut, where the family of Hamiltonians is $\mathcal{F}=\textup{QMC}$. 
Hamiltonians $H\in \textup{QMC}$ are parametrized by a symmetric $n\times n$ interaction matrix $w$ with nonnegative entries, and defined as
\[H = \sum_{i< j\in [n]}w_{ij}H_{ij},\]
where $H_{ij}$ is a projection onto the singlet state $\frac{\ket{01}-\ket{10}}{\sqrt{2}}$ on qubits $i$ and $j$ tensored with identity on the rest. This model is equivalent to the well-known quantum Heisenberg model and is a simple QMA-hard example of a 2-local Hamiltonian. 

A line of work initiated by Gharibian and Parekh \cite{GP} has studied approximation algorithms for this problem \cite{parekh-thompson,anshu2020beyond,parekh2022optimal,lee2022,king2023improved,huber2024,wattssu2,takahashisu2}. The best known results in the literature are an algorithm showing $\alpha^*_{\textup{QMC}}\geq 0.599$ \cite{eunou,jorquera2024monogamy}, and a (conditional) hardness result showing $\alpha^*_{\textup{QMC}}\leq 0.956$ \cite{hwang2023unique}, leaving open a large range of possible values of $\alpha^*_{\textup{QMC}}$.

Our current understanding suggests that Quantum Max Cut presents two key conceptual barriers in closing this gap.
\begin{itemize}
    \item First, the Hamiltonian terms are antiferromagnetic, encouraging pairs of qubits to magnetize in opposite directions on the Bloch sphere. This is akin to the difficulty already found in classical problems like Max Cut.
    \item Second, there is a purely quantum question of finding the correct ground state (or approximate ground state) entanglement structure.
\end{itemize}
Existing algorithmic and hardness techniques seem unlikely to fully capture the interplay between these two obstacles. In order to isolate the second challenge, King \cite{king2023improved} proposed a new simpler Hamiltonian, known as the EPR Hamiltonian.

\paragraph{EPR Hamiltonian.}

The EPR Hamiltonian family $\mathcal{F}=\text{EPR}$ is again parametrized by a symmetric $n\times n$ interaction matrix $w$ with nonnegative entries, and is defined as
\[H = \sum_{i< j\in [n]}w_{ij}E_{ij},\]
where $E_{ij}$ is a projection onto the EPR state $\ket{\EPR}=\frac{\ket{00}+\ket{11}}{\sqrt{2}}$ on qubits $i$ and $j$ tensored with identity on the rest. Notice that the Hamiltonian terms $E_{ij}$ are ferromagnetic; for example, the product state achieving optimal energy for any EPR Hamiltonian is the trivial state $\rho = \ketbra{0^n}{0^n}$. In contrast, for QMC Hamiltonians, computing even a $0.956$-approximation to the optimal product state is computationally hard \cite{hwang2023unique}.

King \cite{king2023improved} gave an algorithm showing $\alpha^*_{\text{EPR}}\geq 1/\sqrt{2}\approx 0.707$, which was later slightly improved to $0.72$ \cite{jorquera2024monogamy}. However, no hardness result are currently known and whether $\alpha^*_{\text{EPR}}$ is strictly less than $1$ is currently open. 

Our main result is an improved approximation algorithm for EPR:
\begin{theorem}\label{thm:main}
        $\alpha^*_{\text{EPR}}\geq \frac{1+\sqrt{5}}{4}\approx 0.809$, i.e., there is a polynomial time $\frac{1+\sqrt{5}}{4}$-approximation algorithm for the EPR problem.
\end{theorem}

\paragraph{Bipartite QMC.}

We call an instance of QMC or EPR \emph{bipartite} if its interaction matrix $w$ corresponds to the adjacency matrix of a bipartite graph. In other words, the $n$ qubits can be partitioned into two disjoint sets $[n] = A\cup B$ such that $w_{ij}=0$ for all pairs $i,j\in A$ and pairs $i,j\in B$, meaning interactions occur only between qubits in different sets. Bipartite instances of QMC and EPR are equivalent up to a Pauli $Y$ applied to every qubit of the subset $A$ (or $B$). As a result, our algorithm also applies to bipartite QMC instances, improving upon the previously best known $3/4$-approximation that one can infer from the work of Lee and Parekh \cite{eunou}.

\begin{corollary}\label{cor:bipartite}
    There is a polynomial time $\frac{1+\sqrt{5}}{4}$-approximation algorithm for bipartite instances of the Quantum Max Cut problem.
\end{corollary}

\section{Algorithm}\label{sec:2}

\subsection{An upper bound on $\lambda_{\max}(H)$}
We begin by describing an efficiently computable upper bound on $\lambda_{\max}(H)$. Let us denote by $\textup{LP}(w)$ the optimal value of the following linear program, known as the \emph{fractional matching LP}:
    \begin{align*}
        \max &\sum_{i<j}w_{ij}\cdot x_{ij}\\
         \text{s.t.}\ \sum_{j\neq i}x_{ij}&\leq 1\quad\text{for all $i\in [n]$}\\
         x_{ij}&\geq 0\quad\text{for all }i,j\in [n]
    \end{align*}
Let $x = (x_{ij})_{i < j}$ be an optimal solution to the above LP. We say that the \emph{LP energy} on edge $\{i,j\}$ is $\tilde{E}_{ij} = \frac{1 + x_{ij}}{2}$.

The \emph{star bound}, proved by Anshu-Gosset-Morenz \cite{anshu2020beyond}, gives a way to bound the ground state energy $\lambda_{\max}(H)$ in terms of the above linear program. We will need a version of this bound that applies to the EPR Hamiltonian (appearing in \cite[Lemma~7]{king2023improved} for example).

\begin{lemma}\label{lem:star-bound}
    For any $n$-qubit state $\rho$ and $i\in [n]$,
    \[\sum_{j\neq i}\max(0,2\tr(\rho\cdot E_{ij})-1)\leq 1.\]
\end{lemma}
As a consequence, $\{\max(0,2\tr(\rho\cdot E_{ij})-1)\}_{i<j}$ always gives a feasible solution to the LP, implying that $\sum_{i<j} w_{ij}\max(0,2\tr(\rho\cdot E_{ij})-1)\leq \textup{LP}(w)$. This implies the upper bound
\begin{equation}\label{eqn:energy-upper-bound}
    \lambda_{\max}(H) =\max_{\rho\succeq 0,\tr(\rho)=1}\tr(\rho\cdot H)\leq \frac{\sum_{i<j}w_{ij} + \textup{LP}(w)}{2}=\sum_{i<j}w_{ij}\tilde{E}_{ij}.
\end{equation}
All the algorithms we present in this paper involve computing an optimal solution $x$ to $\textup{LP}(w)$, and applying a \emph{rounding algorithm} to $x$ to compute the description of a state $\rho$ achieving energy
\begin{equation}\label{eqn:rounding-guarantee}
    \tr(\rho\cdot E_{ij})\geq \alpha\cdot \tilde{E}_{ij}
\end{equation}
on every edge $\{i,j\}\in \binom{[n]}{2}$ for some $\alpha>0$. By \cref{eqn:energy-upper-bound}, we have
\[\tr(\rho\cdot H) = \sum_{i<j}w_{ij}\cdot \tr(\rho\cdot E_{ij})\geq \alpha\cdot \sum_{i<j}w_{ij}\tilde{E}_{ij}\geq \alpha\cdot \lambda_{\max}(H),\] 
so this results in an $\alpha$-approximation algorithm.

Let us describe the rough plan for our presentation of this rounding algorithm. First, in \cref{sec:2-1}, we describe how to use a rounding strategy used in the work of Lee and Parekh \cite{eunou} to achieve $\alpha=3/4$ in the special case that the interaction matrix $w$ is bipartite. Next, in \cref{sec:2-2} we give an improved rounding algorithm achieving $\alpha=\frac{1+\sqrt{5}}{4}$ for bipartite instances. Finally, in \cref{sec:2-3} we show how to achieve the same approximation ratio $\alpha=\frac{1+\sqrt{5}}{4}$ for general non-bipartite instances, proving \cref{thm:main}.

\subsection{A $3/4$-approximation in the bipartite case}\label{sec:2-1}

In the bipartite case, it is well known that the fractional matching LP is integral, i.e., its vertices are indicator vectors of matchings.

\begin{lemma}[\cite{edmonds1965maximum}]\label{lemma:bip-matching}
    For any weighted bipartite graph with adjacency matrix $w$, $\textup{LP}(w) = w(M)$, where $M$ is the maximum weight matching of the graph. That is, the optimal solution to $\textup{LP}(w)$ is achieved by an integer vector $x_{ij} = \mathbf{1}_{\{\{i,j\}\in M\}}$.
\end{lemma}

We can find such an optimal solution using an algorithm for maximum weighted matching. Let $U\subseteq [n]$ be the set of vertices that do not appear in the matching $M$. Consider the two states
\[\rho_{\text{prod}} = \ketbra{0^n}{0^n},\quad \rho_{\text{match}} = \left(\bigotimes_{\{i,j\}\in M} (\ketbra{\EPR}{\EPR})_{i,j}\right)\otimes \left(\bigotimes_{i\in U}\frac{I_i}{2}\right).\]
The rounding algorithm will output a random choice between $\rho_{\text{prod}}$ and $\rho_{\text{match}}$. In particular, we output the state $\rho=(1-p)\cdot \rho_{\text{prod}} + p\cdot \rho_{\text{match}}$ for some choice of $p\in [0,1]$.
\begin{claim}\label{claim:bip-bad-energy}
    The energy attained on edge $\{i,j\}$ is
\[ \tr(\rho \cdot E_{ij})=\begin{cases}
    \frac{1}{2}+\frac{p}{2} &\text{ if }\{i,j\}\in M,\\\frac{1}{2}-\frac{p}{4}\quad &\text{ if }\{i,j\}\notin M.
\end{cases}\]
\end{claim}

On the other hand, the LP energy on edge $\{i,j\}$ is
\[\tilde{E}_{ij}=\frac{1+x_{ij}}{2}=\begin{cases}
    1 &\text{ if }\{i,j\}\in M,\\\frac{1}{2}\quad &\text{ if }\{i,j\}\notin M.
\end{cases}\]
If we set $p=1/2$, then one can check that for each edge $\{i,j\}$,
\[\tr(\rho \cdot E_{ij}) = \frac{3}{4}\tilde{E}_{ij},\]
proving \cref{eqn:rounding-guarantee} for $\alpha=\frac{3}{4}$.

\subsection{An improvement in the bipartite case}\label{sec:2-2}

Our improvement to this algorithm is simple to state -- we interpolate between $\rho_{\text{prod}}$ and $\rho_{\text{match}}$ using quantum superposition rather than in probability. Concretely, let $M$ be the maximum matching and $U$ be the unmatched vertices as in \cref{sec:2-1}. Define the tilted EPR state $\ket{\EPR_\theta}:=\sqrt{\theta}\ket{00}+\sqrt{1-\theta}\ket{11}$. We will output the state
\begin{equation}\label{eqn:bip-state}
    \rho = \left(\bigotimes_{\{i,j\}\in M} (\ketbra{\EPR_\theta}{\EPR_\theta})_{i,j}\right)\otimes \left(\bigotimes_{i\in U}\left(\theta\ketbra{0}{0} + (1-\theta)\ketbra{1}{1}\right)_{i}\right).
\end{equation}

\begin{claim}\label{claim:bip-energy}
The energy attained on edge $\{i,j\}$ is
\[ \tr(\rho \cdot E_{ij})=\begin{cases}
    \frac{1}{2}+\sqrt{\theta(1-\theta)}=\frac{1}{2}+\gamma\quad &\text{ if }\{i,j\}\in M,\\\frac{1}{2}-\theta(1-\theta)=\frac{1}{2}-\gamma^2\quad &\text{ if }\{i,j\}\notin M,
\end{cases}\]
where $\gamma = \sqrt{\theta(1-\theta)}$.
\end{claim}
Comparing \cref{claim:bip-energy} to \cref{claim:bip-bad-energy}, one can readily see that this algorithm performs strictly better; for instance, set $\gamma=p/2$. Our new rounding scheme achieves strictly better energy on edges outside $M$ and identical energy on edges inside $M$.
%
%

It remains to prove \cref{claim:bip-energy}, and then to analyze the approximation ratio of this rounding algorithm. We will repeatedly use the following simple properties of $\ket{\EPR_\theta}$.

\begin{claim}
\label{claim:tiled_energy}
    Let $\gamma = \sqrt{\theta(1-\theta)}$. The following are true.
    \begin{enumerate}
        \item $\abs{\braket{\EPR|\EPR_\theta}}^2 = \frac{1}{2} + \gamma$.
        \item $\tr_{[1]}\ketbra{\EPR_\theta}{\EPR_\theta}=\tr_{[2]}\ketbra{\EPR_\theta}{\EPR_\theta}=\theta\ketbra{0}{0} + (1-\theta)\ketbra{1}{1}$.
        \item $\bra{\EPR} \left(\theta\ketbra{0}{0}+(1-\theta)\ketbra{1}{1}\right)^{\otimes 2} \ket{\EPR} = \frac{1}{2} - \gamma^2$.
    \end{enumerate}
\end{claim}

Now, one can compute all the 2-qubit marginals: 
\[\tr_{[n]\setminus\{i,j\}}(\rho)=\begin{cases}
    \ketbra{\EPR_\theta}{\EPR_\theta}\quad &\text{ if }\{i,j\}\in M,\\\left(\theta\ketbra{0}{0}+(1-\theta)\ketbra{1}{1}\right)^{\otimes 2}\quad &\text{ if }\{i,j\}\notin M,
\end{cases}\]
where for the second case we used that the $1$-qubit marginals are all equal to
$\tr_{[n]\setminus\{1\}}[\ketbra{\EPR_\theta}{\EPR_\theta}]=\theta\ketbra{0}{0}+(1-\theta)\ketbra{1}{1}$.  By \cref{claim:tiled_energy}, the energy achieved by each term $\{i,j\}$ is determined by
\[\bra{\EPR}\tr_{[n]\setminus\{i,j\}}(\rho)\ket{\EPR}=\begin{cases}
    \frac{1}{2}+\sqrt{\theta(1-\theta)}=\frac{1}{2}+\gamma\quad &\text{ if }\{i,j\}\in M,\\\frac{1}{2}-\theta(1-\theta)=\frac{1}{2}-\gamma^2\quad &\text{ if }\{i,j\}\notin M,
\end{cases}\]
and \cref{claim:bip-energy} follows immediately. Now let us analyze the approximation ratio. Similarly to \cref{sec:2-1}, we have the LP energy
\begin{equation}\label{eqn:bip-good-lp-energy}
    \tilde{E}_{ij}=\begin{cases}
    1 &\text{ if }\{i,j\}\in M,\\\frac{1}{2}\quad &\text{ if }\{i,j\}\notin M.
\end{cases}
\end{equation}
Setting $\gamma = \frac{\sqrt{5}-1}{4}$, one can check using \cref{eqn:bip-good-lp-energy} and \cref{claim:bip-energy} that
\[\tr(\rho\cdot E_{ij})= \frac{1+\sqrt{5}}{4}\cdot\tilde{E}_{ij}\]
for each edge $\{i,j\}\in \binom{[n]}{2}$. Thus, we have shown \cref{eqn:rounding-guarantee} with $\alpha=\frac{1+\sqrt{5}}{4}$ as desired.

\subsection{General Case}\label{sec:2-3}

In the general non-bipartite case, we cannot apply \cref{lemma:bip-matching}; the fractional matching LP is not necessarily integral. Previous works proceeded by bounding the integrality gap of this LP, which is quite lossy. Our main idea is to take advantage of the fact that the fractional matching LP is \emph{half-integral}.
\begin{lemma}[{\cite[Theorem 7.5.1]{lovasz2009matching}}]\label{lemma:general-lp-vertices}
    For any weighted $n$-vertex graph with adjacency matrix $w$, there is a vertex-disjoint collection of edges $M\subset \binom{[n]}{2}$ and cycles $\mathcal{C}\subset 2^{\binom{[n]}{2}}$ such that $\textup{LP}(w) = w(M) + \frac{1}{2}\sum_{C\in \mathcal{C}}w(C)$. That is, the optimal solution to $\textup{LP}(w)$ is achieved by the half-integer vector $x_{ij}=\mathbf{1}_{\{\{i,j\}\in M\}} + \frac{1}{2}\sum_{C\in\mathcal{C}}\mathbf{1}_{\{\{i,j\}\in C\}}$.
\end{lemma}

In particular, (e.g. using the algorithm of Anstee \cite{anstee1987polynomial}), one can efficiently find such a solution $x$, i.e. a vertex-disjoint collection of edges $M\subset \binom{[n]}{2}$ and odd length cycles $\mathcal{C}$ such that the LP energies can be written as
\begin{equation}\label{eqn:general-lp-energies}
    \tilde{E}_{ij}=\frac{1+x_{ij}}{2}=\begin{cases}
    1 \quad&\text{ if }\{i,j\}\in M,\\
    \frac{3}{4}\quad &\text{ if }\{i,j\}\in C\text{ for some }C\in \mathcal{C},\\\frac{1}{2}\quad &\text{ otherwise }.
\end{cases}
\end{equation}

Let $U$ be the set of vertices absent from $M$ and $\mathcal{C}$. The algorithm in the bipartite case (\cref{eqn:bip-state}) already suggests what quantum state we will output on the qubits involved in $M$ and $U$. On the qubits in $C$ for some odd cycle $C\in\mathcal{C}$ of length $k$, we will output a high energy state for the EPR Hamiltonian on the length $k$ cycle, i.e. $H_{\text{cycle}}=\sum_{i\in [k]}E_{i, i+1(\text{mod } k)}$, subject to the constraint that all $1$-qubit marginals are equal to $\theta\ketbra{0}{0}+(1-\theta)\ketbra{1}{1}$. The following lemma shows that there is a way to achieve sufficiently large energy on the cycle edges.

\begin{lemma}\label{lem:cycle-state}
    For any integer $k$, there is an efficient algorithm to compute the description of a $k$-qubit state $\rho_k$ satisfying the following, where $\gamma = \frac{\sqrt{5}-1}{4}$ and $\theta\geq 1/2$ is the solution to $\sqrt{\theta(1-\theta)}=\gamma$.
    \begin{enumerate}
        \item\label{item:1} For all $i\in [k]$, $\tr(\rho_k\cdot E_{i,i+1(\textup{mod k})})>\frac{3}{4}\cdot\frac{1+\sqrt{5}}{4}$.
        \item\label{item:2} For all $i\in [k]$, the $1$-qubit marginal $\tr_{[k]-i}[\rho_k]$ equals $\theta\ketbra{0}{0}+(1-\theta)\ketbra{1}{1}$.
        \item\label{item:3} For all $i,j\in [k]$, $\tr(\rho_k\cdot E_{ij})>\frac{1}{2}\cdot \frac{1+\sqrt{5}}{4}$.
    \end{enumerate}
\end{lemma}

We will provide a computer-assisted proof of \cref{lem:cycle-state} later on in \cref{sec:proof-cycle}. We remark that bounds \cref{item:1,item:3} above are loose; for ease of exposition we have chosen to use the weakest bounds that suffice to attain the required approximation ratio.

Our algorithm will output the state
\begin{equation}
    \rho = \left(\bigotimes_{C\in \mathcal{C}}(\rho_{|C|})_{V(C)}\right)\otimes\left(\bigotimes_{\{i,j\}\in M} (\ketbra{\EPR_\theta}{\EPR_\theta})_{i,j}\right)\otimes \left(\bigotimes_{i\in U}\left(\theta\ketbra{0}{0} + (1-\theta)\ketbra{1}{1}\right)_{i}\right),
\end{equation}
where $V(C)\subseteq[n]$ is the set of qubits in the cycle $C$ and $\rho_{|C|}$ is the state guaranteed by \cref{lem:cycle-state}. Similarly to \cref{sec:2-2}, we set $\gamma = \frac{\sqrt{5}-1}{4}$ and $\theta\geq 1/2$ such that $\sqrt{\theta(1-\theta)}=\gamma$. Let us analyze the energy achieved by this algorithm. We have
\[\tr(\rho \cdot E_{ij}) \geq \begin{cases}
    \frac{1+\sqrt{5}}{4}\quad &\text{ if }\{i,j\}\in M,\hfill (\cref{claim:bip-energy})\\
    \frac{3}{4}\cdot \frac{1+\sqrt{5}}{4}\quad &\text{ if }\{i,j\}\in C\text{ for some }C\in \mathcal{C},\hfill\qquad\qquad\qquad\qquad(\cref{lem:cycle-state}, \cref{item:1})\\
    \frac{1}{2}\cdot\frac{1+\sqrt{5}}{4}\quad &\text{ otherwise}.\hfill(\text{1-qubit marginals, }\cref{lem:cycle-state}, \cref{item:3})
\end{cases}\]
To elaborate on the last case, we note that it must be that either qubits $i$ and $j$ are nonadjacent vertices in the \emph{same} cycle $C\in \mathcal{C}$, or they are in tensor product. If the former is true, we use the guarantee of \cref{lem:cycle-state}, \cref{item:3}. Otherwise, we use the fact that all the $1$-qubit marginals are equal to $\theta\ketbra{0}{0} +(1-\theta)\ketbra{1}{1}$ along with the calculations from \cref{sec:2-2} to conclude that the energy is $\frac{1}{2}-\gamma^2=\frac{1+\sqrt{5}}{8}$.

Comparing with \cref{eqn:general-lp-energies}, we immediately get that for every edge $\{i,j\}\in \binom{[n]}{2}$,
\[\tr(\rho\cdot E_{ij})\geq \frac{1+\sqrt{5}}{4}\tilde{E}_{ij},\]
proving \cref{eqn:rounding-guarantee} with $\alpha = \frac{1+\sqrt{5}}{4}$. As a result, we obtain an approximation algorithm for the EPR Hamiltonian with approximation ratio $\frac{1+\sqrt{5}}{4}$, proving \cref{thm:main}.

\printbibliography

@article{jorquera2024monogamy,
  title={Monogamy of Entanglement Bounds and Improved Approximation Algorithms for Qudit Hamiltonians},
  author={Jorquera, Zackary and Kolla, Alexandra and Kordonowy, Steven and Sandhu, Juspreet Singh and Wayland, Stuart},
  journal={arXiv preprint arXiv:2410.15544},
  year={2024}
}

@inproceedings{hwang2023unique,
  title={Unique Games hardness of Quantum Max-Cut, and a conjectured vector-valued Borell's inequality},
  author={Hwang, Yeongwoo and Neeman, Joe and Parekh, Ojas and Thompson, Kevin and Wright, John},
  booktitle={Proceedings of the 2023 Annual ACM-SIAM Symposium on Discrete Algorithms (SODA)},
  pages={1319--1384},
  year={2023},
  organization={SIAM}
}

@article{anshu2020beyond,
  title={Beyond product state approximations for a quantum analogue of max cut},
  author={Anshu, Anurag and Gosset, David and Morenz, Karen},
  journal={arXiv preprint arXiv:2003.14394},
  year={2020}
}

@InProceedings{eunou,
  author =	{Lee, Eunou and Parekh, Ojas},
  title =	{{An Improved Quantum Max Cut Approximation via Maximum Matching}},
  booktitle =	{51st International Colloquium on Automata, Languages, and Programming (ICALP 2024)},
  pages =	{105:1--105:11},
  series =	{Leibniz International Proceedings in Informatics (LIPIcs)},
  ISBN =	{978-3-95977-322-5},
  ISSN =	{1868-8969},
  year =	{2024},
  volume =	{297},
  editor =	{Bringmann, Karl and Grohe, Martin and Puppis, Gabriele and Svensson, Ola},
  publisher =	{Schloss Dagstuhl -- Leibniz-Zentrum fur Informatik},
  address =	{Dagstuhl, Germany},
  URN =		{urn:nbn:de:0030-drops-202482},
  doi =		{10.4230/LIPIcs.ICALP.2024.105}
}

@article{king2023improved,
  title={An improved approximation algorithm for quantum max-cut on triangle-free graphs},
  author={King, Robbie},
  journal={Quantum},
  volume={7},
  pages={1180},
  year={2023},
  publisher={Verein zur F{\"o}rderung des Open Access Publizierens in den Quantenwissenschaften}
}

@article{gp,
  title={Almost optimal classical approximation algorithms for a quantum generalization of Max-Cut},
  author={Gharibian, Sevag and Parekh, Ojas},
  journal={arXiv preprint arXiv:1909.08846},
  year={2019}
}

@article{parekh2022optimal,
  title={An optimal product-state approximation for 2-local quantum hamiltonians with positive terms},
  author={Parekh, Ojas and Thompson, Kevin},
  journal={arXiv preprint arXiv:2206.08342},
  year={2022}
}

@book{lovasz2009matching,
  title={Matching theory},
  author={Lov{\'a}sz, L{\'a}szl{\'o} and Plummer, Michael D},
  volume={367},
  year={2009},
  publisher={American Mathematical Soc.}
}

@article{anstee1987polynomial,
  title={A polynomial algorithm for b-matchings: an alternative approach},
  author={Anstee, Richard P},
  journal={Information Processing Letters},
  volume={24},
  number={3},
  pages={153--157},
  year={1987},
  publisher={Elsevier}
}

@article{edmonds1965maximum,
  title={Maximum matching and a polyhedron with 0, 1-vertices},
  author={Edmonds, Jack},
  journal={Journal of research of the National Bureau of Standards B},
  volume={69},
  number={125-130},
  pages={55--56},
  year={1965}
}

@inproceedings{parekh-thompson,
  
  author = {Parekh, Ojas and Thompson, Kevin},
  
  language = {en},
  
  title = {Application of the Level-2 Quantum Lasserre Hierarchy in Quantum Approximation Algorithms},
  
  publisher = {Schloss Dagstuhl – Leibniz-Zentrum für Informatik},
  
  year = {2021},
  
  copyright = {Creative Commons Attribution 4.0 International license}
}

@misc{takahashisu2,
      title={An SU(2)-symmetric Semidefinite Programming Hierarchy for Quantum Max Cut}, 
      author={Jun Takahashi and Chaithanya Rayudu and Cunlu Zhou and Robbie King and Kevin Thompson and Ojas Parekh},
      year={2023},
      eprint={2307.15688},
      archivePrefix={arXiv},
      primaryClass={quant-ph},
}

@article{wattssu2,
   title={Relaxations and Exact Solutions to Quantum Max Cut via the Algebraic Structure of Swap Operators},
   volume={8},
   ISSN={2521-327X},
   journal={Quantum},
   publisher={Verein zur Forderung des Open Access Publizierens in den Quantenwissenschaften},
   author={Watts, Adam Bene and Chowdhury, Anirban and Epperly, Aidan and Helton, J. William and Klep, Igor},
   year={2024},
   month=may, pages={1352} }

@misc{lee2022,
      title={Optimizing quantum circuit parameters via SDP}, 
      author={Eunou Lee},
      year={2022},
      eprint={2209.00789},
      archivePrefix={arXiv},
      primaryClass={quant-ph},
}

@misc{huber2024,
      title={Second order cone relaxations for quantum Max Cut}, 
      author={Felix Huber and Kevin Thompson and Ojas Parekh and Sevag Gharibian},
      year={2024},
      eprint={2411.04120},
      archivePrefix={arXiv},
      primaryClass={quant-ph},
}
\appendix

\section{Proof of \cref{lem:cycle-state}}\label{sec:proof-cycle}

For the reader's convenience we restate \cref{lem:cycle-state} below. The code for verifying the numerical claims made in the below proof is available at \url{https://github.com/anshnagda/EPR-algorithm}.

\begin{lemma}[Restatement of \cref{lem:cycle-state}]
    For any integer $k$, there is an efficient algorithm to compute the description of a $k$-qubit state $\rho_k$ satisfying the following, where $\gamma = \frac{\sqrt{5}-1}{4}$ and $\theta\geq 1/2$ is the solution to $\sqrt{\theta(1-\theta)}=\gamma$.
    \begin{enumerate}
        \item\label{item:11} For all $i\in [k]$, $\tr(\rho_k\cdot E_{i,i+1(\textup{mod k})})>\frac{3}{4}\cdot\frac{1+\sqrt{5}}{4}$.
        \item\label{item:21} For all $i\in [k]$, the $1$-qubit marginal $\tr_{[k]-i}[\rho_k]$ equals $\theta\ketbra{0}{0}+(1-\theta)\ketbra{1}{1}$.
        \item\label{item:31} For all $i,j\in [k]$, $\tr(\rho_k\cdot E_{ij})>\frac{1}{2}\cdot \frac{1+\sqrt{5}}{4}$.
    \end{enumerate}
\end{lemma}

\begin{proof}
    For $k\leq 5$, we numerically verify the claim by explicitly solving an SDP. 

    Similarly, one can numerically verify that there exists a $5$-qubit state $\psi$ satisfying the following:

    \begin{enumerate}
        \item $\frac{1}{4}\cdot \sum_{i\in [4]}\tr(E_{i,i+1}\cdot \psi)\geq 0.668$.
        \item For all $i,j\in [5]$, $\tr(E_{ij}\cdot \psi)> \frac{1}{2}\cdot \frac{1+\sqrt{5}}{4}$.
        \item For all $i\in [5]$, the $1$-qubit marginal $\tr_{[5]-i}[\psi]$ equals $\theta\ketbra{0}{0}+(1-\theta)\ketbra{1}{1}$.
    \end{enumerate}
    
    Now assume $k\geq 7$. Define the state 
        \[\rho'_k = \ketbra{\EPR_\theta}{\EPR_\theta}^{\otimes \frac{k-5}{2}}\otimes \psi.\]
    Let $\textsf{Shift}_i$ be the unitary that shifts qubits in the $k$-cycle $i$ spots. We will set $\rho_k$ to be the shift-invariant state 
    \[\rho_k = \E_{i\sim [k]}\left[\textsf{Shift}_i\cdot \rho'_k\cdot \textsf{Shift}_i^\dagger\right]\;.\]

    We will verify that $\rho_k$ satisfies the three requirements.

    \begin{enumerate}
        \item Let $i\in [k]$. We can write the energy $\tr(\rho_k\cdot E_{i,i+1(\textup{mod k})})$ as the average energy of a random edge $\{j,j+1(\text{mod }k)\}$ under $\rho'_k$. By definition of $\rho'_k$, one can compute
        \[\tr(\rho'_k\cdot E_{j,j+1(\textup{mod k})})\geq\begin{cases}
            \frac{1+\sqrt{5}}{4} \qquad &j\leq k-5\text{ odd},\\
            \frac{1+\sqrt{5}}{8}\qquad &j\leq k-5\text{ even or }j=k,\\
            0.668\qquad &k-5<j<k.
        \end{cases}\]
        Therefore
        \[\tr(\rho_k\cdot E_{i,i+1(\textup{mod k})}) = \frac{\frac{k-5}{2}\cdot \frac{1+\sqrt{5}}{4} + \frac{k-5}{2}\cdot \frac{1+\sqrt{5}}{8} +  4\cdot 0.668+\frac{1+\sqrt{5}}{8}}{k}\;.\]
        One can verify that $4\cdot 0.668+\frac{1+\sqrt{5}}{8} > 5\cdot \frac{3}{4}\cdot \frac{1+\sqrt{5}}{4}$, implying 
        \[\tr(\rho\cdot E_{i,i+1(\textup{mod k})}) >\frac{(k-5)\cdot\frac{3}{4}\cdot\frac{1+\sqrt{5}}{4}+5\cdot \frac{3}{4}\cdot \frac{1+\sqrt{5}}{4} }{k}=\frac{3}{4}\cdot\frac{1+\sqrt{5}}{4}\;.\]        
        \item It suffices to prove that the $1$-qubit marginals of $\rho'_k$ are equal to $\theta\ketbra{0}{0}+(1-\theta)\ketbra{1}{1}$. For qubits $i\leq k-5$, this follows from \cref{claim:tiled_energy}, and for $i > k-4$, this follows by definition of $\psi$.
        \item Let $i\neq j$. It suffices to prove $\tr(\rho'_k\cdot E_{ij})>\frac{1}{2}\cdot \frac{1+\sqrt{5}}{4}$ the same for $\rho_k'$. We consider three cases.
        \begin{itemize}
            \item If $j=i+1$ for some odd $i\leq k-5$, this is implied by \cref{item:11}.
            \item If $k-5<i,j\leq k$, this is implied by the definition of $\psi$.
            \item Otherwise, the 2-qubit marginal $\tr_{[n]\setminus \{i,j\}}(\rho'_k)$ equals $(\theta \ketbra{0}{0} +(1-\theta)\ketbra{1}{1})^{\otimes 2}$, and this is implied by \cref{claim:tiled_energy}.
        \end{itemize}
    \end{enumerate}
    
\end{proof}

\end{document}